\documentclass[conference]{IEEEtran}
\IEEEoverridecommandlockouts
\usepackage{optidef}
\usepackage{amsmath,amsfonts}
\usepackage{algorithm}
\usepackage{array}
\usepackage[caption=false,font=normalsize,labelfont=sf,textfont=sf]{subfig}
\usepackage{textcomp}
\usepackage{stfloats}
\usepackage{url}
\usepackage{verbatim}
\usepackage{amsmath,amssymb,amsfonts}
\usepackage{mathtools,physics, bbm}
\usepackage{amsthm}
\newtheorem{theorem}{Theorem}
\usepackage{graphicx}
\usepackage{xcolor}
\newtheorem*{remark}{Remark}

\newtheorem{definition}{Definition}
\newtheorem{corollary}{Corollary}

\usepackage{mathrsfs}

\usepackage{cite}
\usepackage{amsmath,amssymb,amsfonts}
\usepackage{algorithmic}
\usepackage{graphicx}
\usepackage{textcomp}
\usepackage{xcolor}

\newcommand{\cH}{\mathcal{H}}

\newcommand{\bi}{\boldsymbol{i}}
\newcommand{\bj}{\boldsymbol{j}}
\newcommand{\bk}{\boldsymbol{k}}
\newcommand{\bx}{\mathbf{x}}

\def\BibTeX{{\rm B\kern-.05em{\sc i\kern-.025em b}\kern-.08em
    T\kern-.1667em\lower.7ex\hbox{E}\kern-.125emX}}

\newcommand{\supp}{\mathrm{supp}}
\newcommand{\wt}{\mathrm{wt}}

\usepackage{lipsum}

\newcommand\blfootnote[1]{%
  \begingroup
  \renewcommand\thefootnote{}\footnote{#1}%
  \addtocounter{footnote}{-1}%
  \endgroup
}

\begin{document}

\title{Bounds for Pure Disjoint $(r,\delta)$-Quantum Locally Recoverable Codes}

\author{%
  \IEEEauthorblockN{Evagoras Stylianou and Holger Boche}
  \IEEEauthorblockA{Chair of Theoretical Information Technology, Technical University of Munich, \\  
  Email: \{evagoras.stylianou@tum.de, boche@tum.de\}
}\vspace{-0.8cm}
}

\maketitle

\begin{abstract}
We study pure disjoint \((r,\delta)\)-quantum locally recoverable codes (qLRCs) without assuming a stabilizer structure. We formulate local Knill--Laflamme conditions for recovery from up to \(\delta-1\) erasures within a recovery block, and introduce blockwise Shor--Laflamme and unitary weight enumerators that capture how error weight is distributed across recovery sets. We establish several properties of these enumerators and use them to derive a Singleton-like bound that strengthens the known bound for disjoint \((r,\delta)\)-qLRCs under a purity assumption, as well as a linear-programming upper bound on the code dimension. These results provide a non-stabilizer, weight-enumerator-based approach to the study of pure disjoint \((r,\delta)\)-qLRCs.
\end{abstract}\vspace{-0.07cm}

\section{Introduction} \vspace{-0.1cm}

\blfootnote{H. B. is also affiliated with the Munich Quantum Valley (MQV) and the Munich Center for Quantum Science and Technology (MCQST). The work of E. Stylianou and H. Boche was partly supported by the German Federal Ministry of Research, Technology and Space (BMFTR) within the national initiative on 6G Communication Systems through the research hub 6G-life under Grants 16KISK002, 16KIS2414 and QUARKS, Grant 16KIS1999,
 the BMBF Quantum Projects QUIET, Grant 16KISQ093, QD-CamNetz, Grant 16KISQ077, and QuaPhySI, Grant 16KIS1598K. H. Boche acknowledges funding by the German Research Foundation as part of Germany’s Excellence Strategy - EXC 2050/2 - Project ID 390696704 - Cluster of Excellence "Centre for Tactile Internet with Human-in-the-Loop" (CeTI) of Technische Universität Dresden. H. Boche was also supported by the LTI TUM collaboration with Princeton through the BMFTR as part of the 6G-Atlantic Bridges project.}
Quantum locally recoverable codes (qLRCs) have recently attracted significant
attention~\cite{golowich2023quantum, sharma2025quantum, luo2025bounds,
galindo2026quantum, xie2025two, li2025improved, cao2025optimal}
as a framework for protecting quantum information under locality constraints. In an
$(r,2)$-qLRC, the erasure of a single qudit can be corrected by accessing at most
$r$ other qudits, thereby avoiding recovery operations that act on the entire system.

The study of qLRCs was initiated by Golowich and Guruswami~\cite{golowich2023quantum},
who introduced $(r,2)$-qLRCs and constructed explicit families via the CSS
construction~\cite{calderbank1998quantum}. They also derived a quantum Singleton-like bound for these codes, which, unlike
its classical counterpart~\cite{gopalan2012locality}, can be strengthened when
the recovery sets are disjoint. Luo \emph{et al.}~\cite{luo2025bounds}
subsequently obtained further bounds for CSS-based $(r,2)$-qLRCs, including
quantum Singleton-like and CM-type bounds. Galindo \emph{et al.}~\cite{galindo2026quantum}
then introduced the more general notion of $(r,\delta)$-qLRCs in the stabilizer
setting and established a correspondence with dual-containing classical
$(r,\delta)$-LRCs, leading in particular to a Singleton-like bound for
stabilizer $(r,\delta)$-qLRCs.  More recently, Li \emph{et al.}~\cite{li2025optimal,li2025improved}
derived further bounds for $(r,2)$-qLRCs. With the exception of the quantum Singleton-like bounds in~\cite{golowich2023quantum}, the known qLRC bounds have largely been obtained through connections with underlying classical codes.

A key technique for deriving upper bounds on code parameters is the use of weight
enumerators and the linear-programming (LP) method~\cite{delsarte1973algebraic}.
In the quantum setting, Shor and Laflamme~\cite{shor1997quantum} introduced the
first quantum weight enumerators, now known as the Shor--Laflamme (SL) weight
enumerators, together with the corresponding quantum MacWilliams identities.
Rains~\cite{rains2002quantum} later introduced the unitary weight enumerators,
which are invariant under code equivalence, and related them both
to correctable erasure patterns and to the SL weight enumerators. Using this connection together with the properties of these enumerators,
Rains~\cite{rainsnonbinary} derived the quantum Singleton bound for nonbinary
quantum codes.

The LP approach to quantum codes goes back to~\cite{shor1997quantum} and was
further refined by Rains~\cite{rains1999shadow} through the introduction of the
shadow enumerator, which provides additional constraints on admissible weight
distributions. Ashikhmin and Litsyu~\cite{ashikhmin1999upper} then formulated
this approach systematically and derived several fundamental upper bounds for
binary quantum codes by choosing auxiliary polynomials satisfying suitable
conditions. Their method was later extended to the nonbinary stabilizer setting
in~\cite{ketkar2006nonbinary}. More recently, Lai and
Ashikhmin~\cite{lai2017linear} generalized this framework to
entanglement-assisted quantum codes by introducing split weight
enumerators~\cite{simonis1995macwilliams}.

Motivated by this line of work, we derive upper bounds on the dimension of pure
disjoint $(r,\delta)$-qLRCs using quantum weight enumerators. We first formulate
local KL conditions and introduce blockwise SL and unitary weight enumerators
that record how error weight is distributed across recovery sets. This framework
yields a Singleton-like bound that strengthens the known bound for general
disjoint $(r,\delta)$-qLRCs. In the special case $\delta=2$, our approach recovers the corresponding stabilizer $(r,2)$-qLRC  bound~\cite{galindo2026quantum} without assuming a stabilizer structure, under the additional assumption of disjoint recovery sets. Finally, we use the same
blockwise enumerator framework to derive an LP upper bound on the code
dimension.

\vspace{-0.15cm}

\section{Preliminaries}\label{sec:prelim}\vspace{-0.15cm}
Let $\mathcal H = (\mathbb C^q)^{\otimes n}$ denote the Hilbert space of $n$ qudits of
local dimension $q$. A quantum code $\mathcal Q$ of length $n$, dimension $K$, and
minimum distance $d$ is a $K$-dimensional subspace of $\mathcal H$, and is denoted by
$((n,K,d))_q$. If $K=q^k$, we write $k=\log_q K$ and denote the code by $[[n,k,d]]_q$.
Let $\mathcal D(\mathcal Q)$ denote the set of density operators whose support is contained
in $\mathcal Q$. For any subset $I\subseteq[n]:=\{1,\dots,n\}$, define the subsystem
$\mathcal H_I := \bigotimes_{i\in I}\mathbb C^q$, and let $I^c := [n]\setminus I$.
A quantum channel $\widetilde{\mathcal N}$ on $\mathcal H$ is \emph{supported on $I$}
if it admits Kraus operators of the form
$\widetilde K_a=K_a\otimes \mathbb I_{I^c}$. Equivalently,
$\widetilde{\mathcal N}=\mathcal N^I\otimes \mathrm{id}_{I^c}$, where
$\mathcal N^I$ is a quantum channel on $\mathcal H_I$ with Kraus operators
$\{K_a\}_a$.\vspace{-0.1cm}

\begin{definition}[$(r,\delta)$-qLRC] \label{def:qLRCs}
A quantum code $\mathcal Q$ is an \emph{$(r,\delta)$-qLRC} if for every
qudit $i\in[n]$ there exists a recovery set $R_i\subseteq[n]$ such that
$i\in R_i$ and $|R_i|\le N:=r+\delta-1$, and for every subset $I\subseteq R_i$ with
$|I|\le \delta-1$, there exists a quantum channel
$\widetilde{\mathcal R}^{R_i}_I$ supported on $R_i$ such that 
\begin{align*}
\widetilde{\mathcal R}^{R_i}_I \circ \widetilde{\mathcal N}^{I}(\rho)
=
\rho,
\qquad
\forall\, \rho\in\mathcal D(\mathcal Q),
\end{align*}
for every quantum channel $\widetilde{\mathcal N}^I$ supported on $I$.
\end{definition}

That is, for each $i\in[n]$, any located error affecting at most $\delta-1$ qudits inside the local group $R_i$ can be corrected by a recovery operation supported on $R_i$.

\begin{definition}[Disjoint $(r,\delta)$-qLRCs] \label{def:qlrcsdis}
A \emph{disjoint $(r,\delta)$-qLRC} is an $(r,\delta)$-qLRC whose recovery sets $\{R_i\}_{i=1}^n$ are either identical or disjoint, that is,
\begin{align*}
R_i \cap R_j = \emptyset \quad \text{or} \quad R_i = R_j,
\qquad \forall\, i,j\in[n].
\end{align*}
The distinct recovery sets can then be labeled as $\{J_\ell\}_{\ell=1}^s$, and they form a partition of $[n]=\bigsqcup_{\ell=1}^s J_\ell$.
\end{definition}

 Disjoint recovery sets induce a blockwise error decomposition that is essential for the weight enumerator arguments.

\begin{remark}
For simplicity, we assume that all recovery blocks have the same size $|J_\ell|=N$ and that $N\mid n$, so that $s=n/N$.
The results extend to unequal block sizes and to $N\nmid n$.
\end{remark}

\subsection{Local structure of disjoint $(r,\delta)$-qLRCs} \label{sec:onb}
For disjoint qLRCs, locality induces a blockwise structure on the code. Let
$\mathcal{Q}$ be a disjoint $(r,\delta)$-qLRC of dimension $K$, and fix an
orthonormal basis (ONB) $\{\ket{\psi_i}\}_{i=1}^K$ of $\mathcal{Q}$. Denote~by
\begin{align*}
P \;:=\; \sum_{i=1}^K \ket{\psi_i}\bra{\psi_i},
\end{align*}
the orthogonal projector onto the code space. Since the recovery sets are disjoint, they induce a partition of the $n$ physical
qudits into $s$ blocks $\cH 
\;=\;
\bigotimes_{\ell=1}^s \mathcal{H}_{J_\ell}$ with 
$\mathcal{H}_{J_\ell} \cong (\mathbb{C}^q)^{\otimes N}$.
\begin{definition}[Local code support]
For each recovery block $J_\ell$, define the \emph{local code support}
\begin{align*}
\mathcal{Q}_\ell
\;:=\;
\operatorname{supp}\bigl(\operatorname{Tr}_{J_\ell^c}(P)\bigr)
\;\subseteq\;
\mathcal{H}_{J_\ell},
\end{align*}
and let $P_\ell$ denote the orthogonal projector onto $\mathcal{Q}_\ell$, referred to
as the \emph{local projector}. Write
$a_\ell := \dim \mathcal{Q}_\ell
\le \dim \mathcal{H}_{J_\ell}$.\vspace{-0.1cm}
\end{definition}
By construction, $\mathcal Q_\ell$ is the smallest subspace of
$\mathcal H_{J_\ell}$ that supports all reduced code states on block $J_\ell$.
Thus, it captures the local degrees of freedom accessible within that block.
Choose, for each $\ell$, an arbitrary ONB
$\{|\alpha_t^{(\ell)}\rangle\}_{t=1}^{a_\ell}$ of $\mathcal Q_\ell$. Then
\begin{align*}
\ket{\alpha_{\mathbf t}}
:=
\bigotimes_{\ell=1}^s |\alpha_{t_\ell}^{(\ell)}\rangle,
\qquad
\mathbf t=(t_1,\dots,t_s)\in\prod_{\ell=1}^s[a_\ell],
\end{align*}
forms an ONB of the product space $\bigotimes_{\ell=1}^s\mathcal Q_\ell$.
Since $\operatorname{supp}(\operatorname{Tr}_{J_\ell^c}(|\psi\rangle\langle\psi|))\subseteq \mathcal Q_\ell$, we have
$|\psi\rangle\in \mathcal Q_\ell\otimes\mathcal H_{J_\ell^c}$ for every $\ell$.
Intersecting over all $\ell$ and using $\mathcal H=\bigotimes_{\ell=1}^s\mathcal H_{J_\ell}$ yields
$|\psi\rangle\in\bigotimes_{\ell=1}^s \mathcal Q_\ell$. Hence, \vspace{-0.15cm}
% every 
% $\ket{\psi_i}\in \mathcal{Q}\subseteq
% \bigotimes_{\ell=1}^s \mathcal{Q}_\ell$ admits the decomposition\vspace{-0.1cm}
\begin{align}
\ket{\psi}
\;=\;
\sum_{\mathbf t}
\nu_{\mathbf t}^{(\ket{\psi})} \,\ket{\alpha_{\mathbf t}},
\quad \forall \ket{\psi}\in \mathcal{Q}\subseteq
\bigotimes_{\ell=1}^s \mathcal{Q}_\ell .\label{eq:decomp}
\end{align}
In general, this inclusion is strict, since local supports do not capture inter-block correlations.

To define the quantum weight enumerators, we fix an orthonormal operator basis. Let $\{E_j\}_{j=0}^{q^2-1}$ be an operator basis of $\mathbb C^q$ satisfying
$\Tr \big (E_j^\dagger E_k\big )=q\,\delta_{j,k}$, with $E_0=\mathbb I$. This induces the tensor-product operator basis
$E_{\boldsymbol i}=E_{i_1}\otimes\cdots\otimes E_{i_n}$ on
$(\mathbb C^q)^{\otimes n}$, where
$\boldsymbol i\in\{0,\dots,q^2-1\}^n$, and
$\Tr\big (E_{\boldsymbol i}^\dagger E_{\boldsymbol j}\big)=q^n\delta_{\boldsymbol i,\boldsymbol j}$.
For $E=F_1\otimes\cdots\otimes F_n$, define
$\supp(E):=\{j\in[n]:F_j\neq\mathbb I\}$ and $\wt(E):=|\supp(E)|$.

\subsection{Krawtchouk polynomials and expansion}\vspace{-0.05cm}
To analyze blockwise quantum weight enumerators and formulate LP--based bounds, we use Krawtchouk polynomials and their multivariate generalizations. For an alphabet of size $q^2$, the Krawtchouk polynomials of degree $0 \le i \le n$ in the variable $0 \le k \le n$ are defined by \vspace{-0.05cm}
\begin{align*}
K^{(n)}_i(k)
=
\sum_{j=0}^i
(-1)^j (q^2-1)^{i-j}
\binom{k}{j}\binom{n-k}{i-j}.
\end{align*}
These polynomials satisfy the discrete orthogonality relation $\sum_{j=0}^n K^{(n)}_i(j)\,K^{(n)}_j(m)
=
q^{2n}\,\delta_{im}.$
Hence, they form a basis for the space of
polynomials of degree at most $n$. Let $n=sN$ where $s$ and $N$ are positive integers, and
write $\bi=(i_1,\dots,i_s)$ and $\bk=(k_1,\dots,k_s)$ with
$i_\ell,k_\ell\in\mathcal I_N:=\{0,\dots,N\}$. The multivariate Krawtchouk
polynomials are defined by
$K_{\bi}(\bk)=\prod_{\ell=1}^s K_{i_\ell}^{(N)}(k_\ell)$. Then the family
$\{K_{\bi}\}_{\bi\in\mathcal I_N^s}$ forms a basis for the space of multivariate
polynomials in $s$ variables with degree at most $N$ in each coordinate. Hence, any such polynomial $f(\bk)$ admits a unique Krawtchouk expansion
\vspace{-0.05cm}
\begin{align}
f(\bk)
=
\sum_{\bi\in\mathcal{I}_N^s} f_{\bi}\,K_{\bi}(\bk),  \label{eq:expansion}
\end{align}
where $f_{\bi}
= q^{-2n}
\sum_{\bk\in\mathcal{I}_N^s}
f(\bk)\,K_{\bk}(\bi).$ 
% \begin{align*}
% f_{\bi}
% =
% \frac{1}{q^{2n}}
% \sum_{\bk\in\mathcal{I}_N^s}
% f(\bk)\,K_{\bk}(\bi).
% \end{align*}

\section{KL Conditions for Disjoint qLRCs}
In this section, we derive local KL conditions for disjoint
$(r,\delta)$-qLRCs, since the standard KL conditions do not encode the locality constraint on the recovery operation.

\begin{theorem}[Local KL condition]
\label{thm:local-KL}
Let $\mathcal Q$ be a disjoint $(r,\delta)$-qLRC with partition
$\{J_\ell\}_{\ell=1}^s$. Fix a block $J_\ell$ and let $P_\ell$ be the local
projector. Consider a quantum channel $\widetilde{\mathcal N}_\ell$ supported on
$J_\ell$, with Kraus operators $\{K_a^{(\ell)}\otimes \mathbb I_{J_\ell^c}\}_a$.
Then there exists a quantum channel $\widetilde{\mathcal R}_\ell$ supported on
$J_\ell$ such that
\begin{align*}
(\widetilde{\mathcal R}_\ell \circ \widetilde{\mathcal N}_\ell)(\rho)
=
\rho,
\qquad
\forall\, \rho \in \mathcal{D}(\mathcal Q),
\end{align*}
if and only if
\begin{align*}
P_\ell (K_a^{(\ell)})^\dagger K_b^{(\ell)} P_\ell
=
\widetilde C^{(\ell)}_{ab}\, P_\ell,
\qquad
\forall\, a,b,
\end{align*}
for some scalars $\widetilde C^{(\ell)}_{ab}\in\mathbb C$. Equivalently,
\begin{align*}
\langle \alpha_t^{(\ell)}|
(K_a^{(\ell)})^\dagger K_b^{(\ell)}
|\alpha_{t'}^{(\ell)}\rangle
=
\widetilde C^{(\ell)}_{ab}\,\delta_{tt'},
\qquad
\forall\, a,b,\ t,t'\in[a_\ell],
\end{align*}
where $\{|\alpha_t^{(\ell)}\rangle\}_{t=1}^{a_\ell}$ is an ONB of $\mathcal Q_\ell$.
\end{theorem}
\begin{proof}[Proof sketch]
Fix a block $J_\ell$ and suppress the superscript $(\ell)$. Write
$\widetilde{\mathcal N}_\ell=\mathcal N_\ell\otimes \mathrm{id}_{J_\ell^c}$, with Kraus operators $\{K_a\otimes \mathbb I_{J_\ell^c}\}_a$. Let $\{|\alpha_t\rangle\}_{t=1}^{a_\ell}$ be an ONB of the local code support
$\mathcal Q_\ell$. Then every code basis vector $|\psi_i\rangle\in\mathcal Q$ can be
written as
$|\psi_i\rangle=\sum_{t=1}^{a_\ell}|\alpha_t\rangle\otimes |\beta_{t,i}\rangle$,
with $\sum_t\langle\beta_{t,i}|\beta_{t,j}\rangle=\delta_{ij}$.

\emph{Sufficiency.}
If $\langle \alpha_t|K_a^\dagger K_b|\alpha_{t'}\rangle
=\widetilde C_{ab}\delta_{tt'}$ for all $a,b,t,t'$, then this is exactly the usual KL condition on the local code support $\mathcal Q_\ell$. Hence, by the standard KL theorem, there exists a recovery map $\mathcal R_\ell$ on $\mathcal H_{J_\ell}$ correcting $\mathcal N_\ell$ on $\mathcal Q_\ell$. Extending it by the identity on $J_\ell^c$ yields a recovery $\widetilde{\mathcal R}_\ell=\mathcal R_\ell\otimes \mathrm{id}_{J_\ell^c}$ supported on $J_\ell$ that corrects $\widetilde{\mathcal N}_\ell$ on $\mathcal Q$.

\emph{Necessity.}
Conversely, suppose there exists a recovery
$\mathcal R_\ell\otimes \mathrm{id}_{J_\ell^c}$ such that
$((\mathcal R_\ell\circ\mathcal N_\ell)\otimes \mathrm{id}_{J_\ell^c})(\rho)=\rho$
for all $\rho\in\mathcal D(\mathcal Q)$. By considering a purification of the maximally mixed state on $\mathcal Q$, one
shows that $\mathcal R_\ell\circ\mathcal N_\ell$ acts as the identity on the
local support $\mathcal Q_\ell$.
Hence, $\mathcal N_\ell$ is perfectly correctable on the local code
$\mathcal Q_\ell\subseteq\mathcal H_{J_\ell}$, and the standard
KL necessity theorem yields
$P_\ell K_a^\dagger K_b P_\ell=\widetilde C_{ab}P_\ell$ for all $a,b$.
% \emph{Necessity.}
% Conversely, suppose there exists a recovery
% $\mathcal R_\ell\otimes \mathrm{id}_{J_\ell^c}$
% supported on $J_\ell$ such that
% $((\mathcal R_\ell\circ\mathcal N_\ell)\otimes \mathrm{id}_{J_\ell^c})(\rho)=\rho$
% for all $\rho\in\mathcal D(\mathcal Q)$. Since the noise and recovery act
% nontrivially only on $\mathcal H_{J_\ell}$, the induced channel
% $\mathcal R_\ell\circ\mathcal N_\ell$ acts trivially on the local support
% $\mathcal Q_\ell$. Equivalently,
% $\mathcal N_\ell$ is perfectly correctable on the code $\mathcal Q_\ell\subseteq
% \mathcal H_{J_\ell}$. Applying the standard KL necessity theorem to
% $\mathcal Q_\ell$ therefore yields
% $P_\ell K_a^\dagger K_b P_\ell=\widetilde C_{ab}P_\ell,\;\forall a,b$.
\end{proof}\vspace{-0.05cm}

% The local KL condition characterizes \emph{perfect correctability}
% of arbitrary quantum channels supported on a recovery block and depends only on
% the local code support.
The local KL condition shows that correctability within a recovery block depends only on the local code support $\mathcal Q_\ell$. \vspace{-0.05cm}

\begin{remark}
Since the local KL conditions are linear in the error operators, it is enough to work with the tensor-product operator basis $\mathcal E=\{E_{\bi}\}_{\bi}$ of $(\mathbb C^q)^{\otimes n}$ introduced in Section~\ref{sec:onb}.
\end{remark}

The \emph{global minimum distance} $d$ of $\mathcal Q$ is the smallest weight of an
operator $E\in\mathcal E$ that is not detectable on $\mathcal Q$, that is,
$\langle \psi_i|E|\psi_j\rangle \neq C_E\,\delta_{ij}$ for some $i,j$. Likewise, for a
recovery block $J_\ell$ with local code support $\mathcal Q_\ell$, the \emph{local
minimum distance} $\delta^{(\ell)}$ is the smallest weight of an operator
$E^{(\ell)}$ supported on $J_\ell$ that is not locally detectable on
$\mathcal Q_\ell$, that is,
$\langle \alpha_t^{(\ell)}|E^{(\ell)}|\alpha_{t'}^{(\ell)}\rangle \neq
C_E^{(\ell)}\,\delta_{tt'}$ for some $t,t'\in[a_\ell]$. The \emph{local distance} of
the code is then defined by $\delta:=\min_\ell \delta^{(\ell)}$.

Let $\mathcal Q$ be a disjoint $(r,\delta)$-qLRC with global distance $d$ and block
partition $\{J_\ell\}_{\ell=1}^s$. For a basis operator
$E=\bigotimes_{\ell=1}^s E^{(\ell)}$, let $i_\ell=\wt(E^{(\ell)})$ denote the number of affected qudits in block $J_\ell$, and write
$\boldsymbol i\in\mathcal I_N^s$ for the corresponding erasure pattern. By locality and disjointness, each block with $i_\ell\le\delta-1$ can be
corrected independently, while the remaining erasures can be corrected globally
provided their total weight is at most $d-1$. Therefore the set of correctable
erasure patterns is 
\begin{align*}
\mathcal{T}_{d,\delta}
=
\Big\{
\boldsymbol i\in\mathcal{I}_N^s \;\Big|\;
&\exists\, V\subseteq[s]\ \text{s.t.}\ 
\sum_{j\in V} i_j \le d-1,\\
& i_j\le\delta-1 \ \forall\, j\notin V
\Big\}.
\end{align*}
Equivalently, $\mathcal T_{d,\delta}$ consists of all
$\bi\in\mathcal I_N^s$ such that $\sum_{j:\,i_j\ge\delta} i_j \le d-1$.\vspace{-0.05cm}

\section{Blockwise SL and Unitary Weight Enumerators}
Here we introduce blockwise SL and unitary weight enumerators, which record how error weight is distributed across recovery blocks. 
Related refinements include split weight enumerators for classical codes~\cite{simonis1995macwilliams} and two-block split weight enumerators for entanglement-assisted quantum codes~\cite{lai2017linear}.
\vspace{-0.1cm}

\subsection{Blockwise SL weight enumerators}\vspace{-0.05cm}
Using the operator basis $\mathcal E=\{E\}$ fixed in Section~\ref{sec:onb}, write
$E=\bigotimes_{\ell=1}^s E^{(\ell)}$. For $\bi\in\mathcal I_N^s$ and operators $M_1,M_2$ on
$(\mathbb C^q)^{\otimes n}$, define the blockwise SL weight enumerators
\allowdisplaybreaks
\begin{align*}
A^{\mathrm{SL}}_{\bi}(M_1,M_2)
&:= \sum_{\substack{E\in\mathcal E\\ \wt(E^{(\ell)})=i_\ell\ \forall \ell}}
\Tr(EM_1)\,\Tr(E^\dagger M_2),\\
B^{\mathrm{SL}}_{\bi}(M_1,M_2)
&:= \sum_{\substack{E\in\mathcal E\\ \wt(E^{(\ell)})=i_\ell\ \forall \ell}}
\Tr(EM_1E^\dagger M_2).
\end{align*}
Summing over all $\bi$ such that $\sum_\ell i_\ell=i$ recovers the standard
weight-$i$ SL enumerators. When $M_1=M_2=P$, we write  $A^{\mathrm{SL}}_{\bi}$ and $B^{\mathrm{SL}}_{\bi}$.
We next state their basic properties. 
\vspace{-0.1cm}

\begin{theorem}[Properties of the blockwise SL weight enumerators]
\label{thm:LRCweight}
Let $\mathcal Q$ be a disjoint $(r,\delta)$-qLRC with parameters $((n,K,d))_q$, projector $P$, and blockwise SL weight enumerators
$\{A^{\mathrm{SL}}_{\bi}\}$ and $\{B^{\mathrm{SL}}_{\bi}\}$. Then:
\begin{enumerate}
\itemsep0.3em
\item $A^{\mathrm{SL}}_{\bi}\ge 0$ and
$K B^{\mathrm{SL}}_{\bi}\ge A^{\mathrm{SL}}_{\bi}$ for all
$\bi\in\mathcal I_N^s$. Moreover,
$A^{\mathrm{SL}}_{\mathbf 0}=K^2$ and $B^{\mathrm{SL}}_{\mathbf 0}=K$.

\item $A^{\mathrm{SL}}_{\bi}=K\,B^{\mathrm{SL}}_{\bi}$ for all
$\bi\in\mathcal T_{d,\delta}$.

\item They satisfy the MacWilliams identity \vspace{-0.1cm}
\begin{align*}
B^{\mathrm{SL}}_{\bj}
=
\frac{1}{q^n}\sum_{\bi\in\mathcal I_N^s}
A^{\mathrm{SL}}_{\bi}\,K_{\bj}(\bi).
\end{align*}
\end{enumerate}
\end{theorem}
\begin{proof}
1) Nonnegativity is immediate from the definitions, and
$K B^{\mathrm{SL}}_{\bi}\ge A^{\mathrm{SL}}_{\bi}$ follows from the
Cauchy--Schwarz inequality. For $\bi=\mathbf 0$, only the identity operator
contributes, so
$A^{\mathrm{SL}}_{\mathbf 0}=\Tr(P)^2=K^2$ and
$B^{\mathrm{SL}}_{\mathbf 0}=\Tr(P^2)=K$.

2) Fix $\bi\in \mathcal{T}_{d,\delta}$ and let
$E=\bigotimes_{\ell=1}^s E^{(\ell)}$ have block weight distribution $\bi$.
Set $V:=\{\ell\in[s]:\, i_\ell\ge \delta\}$. Then for each $\ell\notin V$ we have
$i_\ell\le\delta-1$, so by the local KL conditions on $\mathcal Q_\ell$,
$\langle \alpha_t^{(\ell)}|E^{(\ell)}|\alpha_{t'}^{(\ell)}\rangle
= C^{(\ell)}\delta_{tt'}$. Using the decomposition in \eqref{eq:decomp}, it follows that
for all $|\psi_a\rangle,|\psi_b\rangle\in\mathcal Q$,\vspace{-0.05cm}
\begin{align*}
\langle\psi_a|E|\psi_b\rangle
=
\Big(\prod_{\ell\notin V}C^{(\ell)}\Big)
\langle\psi_a|E_V|\psi_b\rangle,
\end{align*}
where $E_V^{(\ell)}=E^{(\ell)}$ for $\ell\in V$ and
$E_V^{(\ell)}=\mathbb I$ for $\ell\notin V$. Since $\bi\in\mathcal T_{d,\delta}$, we have $\sum_{\ell\in V} i_\ell\le d-1$ and hence
$\wt(E_V)\le d-1$. By the global KL conditions, $\langle\psi_a|E_V|\psi_b\rangle=C_{\mathrm{global}}\delta_{ab}$, so \vspace{-0.05cm}
\begin{align*}
\langle\psi_a|E|\psi_b\rangle
=
\Big(\prod_{\ell\notin V}C^{(\ell)}\Big)C_{\mathrm{global}}\delta_{ab}.
\end{align*}
Thus every operator with block weight distribution $\bi$ acts as a scalar on
$\mathcal Q$, which implies
$A^{\mathrm{SL}}_{\bi}=K\,B^{\mathrm{SL}}_{\bi}$ for all
$\bi\in\mathcal T_{d,\delta}$.

3) This follows by adapting the arguments in \cite{lai2017linear} to more than two blocks and is omitted due to space constraints.
\end{proof}\vspace{-0.15cm}

% \begin{proof}
% Let $E = \bigotimes_{\ell=1}^s E_\ell$ be an error with block weight $\bi \in \mathcal{T}_{d,\delta}$.  
% Define $V \subseteq [s]$ as the set of blocks whose errors exceed the local distance $\delta-1$, the remaining blocks $[s]\setminus V$ are locally correctable.  
% By the local KL conditions, each $\ell \notin V$ contributes a constant $C^{(\ell)}$, independent of the code state. For the blocks in $V$, the joint error supported on $\cup_{\ell\in V}J_\ell$ has total
% weight at most $d-1$, so the global KL conditions apply. Hence, for any code states $\ket{\psi_i}, \ket{\psi_j}$,
% \begin{align*}
% \bra{\psi_i} E \ket{\psi_j} = \Big(\prod_{\ell \notin V} C^{(\ell)}\Big) C_{\mathrm{global}} \, \delta_{i,j}.
% \end{align*}
% Substituting into the definitions of $A_{\bi}$ and $B_{\bi}$ shows $KB_{\bi} = A_{\bi}$ for all $\bi \in \mathcal{T}_{d,\delta}$, while $KB_{\bi} \ge A_{\bi}$ follows from Cauchy--Schwarz. The McWilliams identity then follows as in standard quantum codes, with sums factorizing over blocks~\cite{lai2017linear}.
% \end{proof}

\subsection{Blockwise Unitary Weight Enumerators}
For $S\subseteq[n]$, define the unitary weight
enumerators~\cite{rains2002quantum}. \allowdisplaybreaks
\begin{align*}
A^{\mathrm U}_S(M_1,M_2)
&:= \Tr\bigl(\Tr_{S^c}(M_1)\,\Tr_{S^c}(M_2)\bigr),\\
B^{\mathrm U}_S(M_1,M_2)
&:= \Tr\bigl(\Tr_S(M_1)\,\Tr_S(M_2)\bigr).
\end{align*}
Thus $A^{\mathrm U}_S(M_1,M_2)=B^{\mathrm U}_{S^c}(M_1,M_2)$. Let $S_\ell:=S\cap J_\ell$. For $\bi\in\mathcal I_N^s$, define the blockwise unitary enumerators by\vspace{-0.05cm}
\begin{align}
A^{\mathrm{U}}_{\bi}(M_1,M_2)
&:= \sum_{\substack{S\subseteq[n]\\ |S_\ell|=i_\ell,\ \forall \ell}}
A^{\mathrm U}_S(M_1,M_2), \label{eq:u1}\\
B^{\mathrm{U}}_{\bi}(M_1,M_2)
&:= \sum_{\substack{S\subseteq[n]\\ |S_\ell|=i_\ell,\ \forall \ell}}
B^{\mathrm U}_S(M_1,M_2). \label{eq:u2}
\end{align}
Summing over all $\bi$ with $\sum_\ell i_\ell=i$ recovers the standard weight-$i$
unitary enumerators. Moreover,
$A^{\mathrm U}_{\bi}(M_1,M_2)=B^{\mathrm U}_{N-\bi}(M_1,M_2)$, where
$N-\bi:=(N-i_1,\dots,N-i_s)$.

To relate the blockwise unitary and SL enumerators, let
$T\subseteq[n]$ decompose as $T=\bigcup_\ell T_\ell$ with $T_\ell\subseteq J_\ell$, and
define
\begin{align*}
A^{\mathrm{SL}}_T(M_1,M_2)
&:= \sum_{\substack{E\in\mathcal E\\ \supp(E^{(\ell)})=T_\ell,\ \forall \ell}}
\Tr(EM_1)\Tr(E^\dagger M_2),
\end{align*}
with $B^{\mathrm{SL}}_T(M_1,M_2)$ defined analogously. Then, following
\cite{rains2002quantum}, for any $S\subseteq[n]$,
\begin{align*}
A^{\mathrm U}_S(M_1,M_2)
=
q^{-|S|}
\sum_{T_\ell\subseteq S_\ell,\ \forall \ell}
A^{\mathrm{SL}}_T(M_1,M_2),
\end{align*}
and the same relation holds for $B^{\mathrm U}_S$ with $B^{\mathrm{SL}}_T$.

Next, we express the blockwise unitary enumerators in terms of the blockwise SL enumerators.
\begin{corollary} \label{cor:enum}
For any operators $M_1,M_2$ on $(\mathbb C^q)^{\otimes n}$,
\begin{align*}
A^{\mathrm U}_{\bi}(M_1,M_2)
=
q^{-\sum_{\ell=1}^s i_\ell}
\sum_{\bk\le \bi}
c(\bk,\bi)\,A^{\mathrm{SL}}_{\bk}(M_1,M_2),
\end{align*}
where $\bk\le\bi$ denotes componentwise inequality and
$c(\bk,\bi):=\prod_{\ell=1}^s\binom{N-k_\ell}{N-i_\ell}$.
An analogous identity holds for $B^{\mathrm U}_{\bi}(M_1,M_2)$.
\end{corollary}

\begin{proof}
This follows by counting, in each block $J_\ell$, the number of subsets
$S_\ell\supseteq T_\ell$ with $|S_\ell|=i_\ell$ for a fixed subset $T_\ell$ of
size $k_\ell$, and then multiplying over $\ell$.
\end{proof}

% \begin{corollary} \label{cor:enum}
% For any operators $M_1,M_2$ on $(\mathbb C^q)^{\otimes n}$,
% \begin{align*}
% A^{\mathrm{U}}_{\boldsymbol i}(M_1,M_2)
% &=
% q^{-\sum_{\ell=1}^s i_\ell}
% \sum_{\boldsymbol k\le\boldsymbol i}
% c(\boldsymbol k,\boldsymbol i)
% A^{\mathrm{SL}}_{\boldsymbol k}(M_1,M_2),
% \end{align*}
% where $\boldsymbol k \le \boldsymbol i$ denotes componentwise inequality and
% \begin{align*}
% c(\boldsymbol k,\boldsymbol i) := \prod_{\ell=1}^s \binom{N-k_\ell}{N-i_\ell}.
% \end{align*}
% An analogous relation holds for $B^{\mathrm{U}}_{\boldsymbol i}$.
% \end{corollary}

% \begin{proof}
% This follows by counting, independently in each block $J_\ell$, the number of choices of $S_\ell \supseteq T_\ell$ with $|S_\ell|=i_\ell$ containing
% a fixed $T_\ell$ of size $k_\ell$
% and grouping terms accordingly.
% \end{proof}

% We now relate $A^{\mathrm{U}}_{\boldsymbol i}$ and
% $B^{\mathrm{U}}_{\boldsymbol i}$ for a $(r,\delta)$-qLRC.
\begin{theorem} \label{thm:unitary-ineq}
Let $\mathcal Q$ be a $((n,K,d))_q$ quantum code with projector $P$. Then, for every
$\bi\in\mathcal I_N^s$,
\begin{align*}
K\,B^{\mathrm U}_{\bi}\ge A^{\mathrm U}_{\bi}.
\end{align*}
Equality holds if and only if $\mathcal Q$ can perfectly correct the erasure of every
subset $S\subseteq[n]$ such that $|S_\ell|\le i_\ell,\;\forall \ell\in[s]$.
\end{theorem}

\begin{proof}
For every subset $S\subseteq[n]$, the unitary enumerators satisfy
$K\,B^{\mathrm U}_S\ge A^{\mathrm U}_S$. Summing over all $S$ such that
$|S_\ell|=i_\ell$ for every $\ell$ gives
$K\,B^{\mathrm U}_{\bi}\ge A^{\mathrm U}_{\bi}$.
For the equality statement, \eqref{eq:u1} and \eqref{eq:u2} give
\begin{align*}
K\,B^{\mathrm U}_{\bi}-A^{\mathrm U}_{\bi}
=
\sum_{\substack{S\subseteq[n]\\ |S_\ell|=i_\ell,\ \forall \ell}}
\bigl(K\,B^{\mathrm U}_S-A^{\mathrm U}_S\bigr).
\end{align*}
Each summand is nonnegative, so equality holds if and only if
$K\,B^{\mathrm U}_S=A^{\mathrm U}_S$ for every $S$ with $|S_\ell|=i_\ell$ for all
$\ell$. By Rains' equality characterization~\cite{rains2002quantum}, this is
equivalent to perfect correctability of the erasure of each such subset $S$. Since
erasure correction is monotone under inclusion, this is in turn equivalent to perfect
correctability of every subset $S'\subseteq[n]$ satisfying $|S'_\ell|\le i_\ell$ for
all $\ell$. The converse is immediate.
\end{proof}
% \begin{theorem} \label{thm:unitary-ineq}
% Let $\mathcal Q$ be a quantum code of dimension $K$ with projector $P$. Then, for all $\boldsymbol i \in \mathcal{I}_N^s$,
% \begin{align*}
% K\, B^{\mathrm{U}}_{\boldsymbol i} \ge A^{\mathrm{U}}_{\boldsymbol i}.
% \end{align*}
% Equality holds if and only if $\mathcal Q$ can perfectly recover from the erasure of every subset $S \subseteq [n]$ satisfying $|S_\ell| \le i_\ell,\:\forall \ell \in [s]$.
% \end{theorem}

% \begin{proof}
% For every subset $S\subseteq[n]$, Rains' unitary enumerators satisfy
% $K\,B^{\mathrm U}_S\ge A^{\mathrm U}_S$.
% Summing over all $S$ such that $|S_\ell|=i_\ell$ for all $\ell$ yields
% $K\,B^{\mathrm{U}}_{\boldsymbol i}\ge A^{\mathrm{U}}_{\boldsymbol i}$. For the equality statement, write using \eqref{eq:u1} and \eqref{eq:u2}, 
% \begin{align*}
% K\,B^{\mathrm{U}}_{\boldsymbol i} -A^{\mathrm{U}}_{\boldsymbol i}
% =\sum_{\substack{S\subseteq[n]\\ |S_\ell|=i_\ell\ \forall \ell}}
% \Bigl(K\,B^{\mathrm U}_S-A^{\mathrm U}_S\Bigr).
% \end{align*}
% Each summand is nonnegative, hence the left-hand side is zero if and only if
% $K\,B^{\mathrm U}_S=A^{\mathrm U}_S$ for every $S$ with $|S_\ell|=i_\ell$.
% By Rains' equality characterization \cite{rains2002quantum}, this condition is
% equivalent to perfect correctability of the erasure of each such subset $S$.
% Since correctability is monotone under inclusion, this is equivalent to correctability for all
% $S'$ with $|S'_\ell|\le i_\ell$.
% The converse is immediate.
% \vspace{-0.1cm}
% \end{proof}

\begin{remark}
The inequality holds for arbitrary quantum codes; locality will be used later to
interpret the equality condition.
\end{remark}

\section{Singleton-like bound for pure disjoint qLRCs}\vspace{-0.1cm}

Next, we derive a Singleton-like bound for pure disjoint $(r,\delta)$-qLRCs using blockwise weight enumerators.\vspace{-0.1cm}

\begin{definition}[Pure disjoint $(r,\delta)$-qLRC] \label{def:pure}
A disjoint $(r,\delta)$-qLRC with distance $d$ is pure if
$A^{\mathrm{SL}}_{\boldsymbol i} = 0$ for all
$\boldsymbol i \in \mathcal{T}_{d,\delta}\setminus\{\boldsymbol 0\}$.
\end{definition}

This is the blockwise analogue of the usual purity condition, restricted to the set of $(r,\delta)$-qLRC-correctable error patterns.\vspace{-0.15cm}

\begin{theorem}[Singleton-like bound for pure disjoint $(r,\delta)$-qLRCs] \label{thm:puresing}
Let $\mathcal Q$ be a pure disjoint
$(r,\delta)$-qLRC of length $n$ and minimum distance $d$. Then \vspace{-0.05cm}
\begin{align}
K \;\le\;
q^{\,n - 2(d-1) - 2\left\lfloor \frac{n-(d-1)}{N} \right\rfloor(\delta-1)}. \label{eq:singbound}
\end{align}
\end{theorem}

\begin{proof}
Set $f:=\left\lfloor \frac{n-(d-1)}{N}\right\rfloor$. Choose $f$ recovery blocks,
say $J_1,\dots,J_f$, and in each select a subset $T_\ell\subseteq J_\ell$ of size
$|T_\ell|=\delta-1$. Let $T:=\bigcup_{\ell=1}^f T_\ell$. Since
$n-fN\ge d-1$, we may also choose a subset
$V\subseteq [n]\setminus \bigcup_{\ell=1}^f J_\ell$ of size $|V|=d-1$.
Let $\bj$ be the corresponding blockwise erasure profile, so that
$j_\ell=\delta-1$ for $\ell\le f$ and $\sum_{\ell>f} j_\ell=d-1$. Set
$m:=\sum_{\ell=1}^s j_\ell=f(\delta-1)+(d-1)$.

Now let $S\subseteq[n]$ satisfy $|S_\ell|\le j_\ell$ for all $\ell$. For
$\ell\le f$, we have $|S_\ell|\le \delta-1$, so the erasures in those blocks are
locally correctable. After those local recoveries, the remaining erasures lie in the
blocks $\ell>f$, where their total number is at most
$\sum_{\ell>f} j_\ell=d-1$. These are therefore globally correctable. Hence every
subset $S$ with block profile bounded by $\bj$ is correctable. Therefore, by Theorem~\ref{thm:unitary-ineq}, we have
$K B^{\mathrm U}_{\bj}=A^{\mathrm U}_{\bj}$.
Using the duality relation $B^{\mathrm U}_{\bj}=A^{\mathrm U}_{N-\bj}$, we obtain 
\begin{align}
A^{\mathrm U}_{\bj}=K A^{\mathrm U}_{N-\bj}. \label{eq:equality}
\end{align}
Expanding both sides using Corollary~\ref{cor:enum}, we obtain
\begin{align*}
A^{\mathrm U}_{\bj}
&= q^{-m}\sum_{\bk\le\bj} c(\bk,\bj)\,A^{\mathrm{SL}}_{\bk},\\
A^{\mathrm U}_{N-\bj}
&= q^{-n+m}\sum_{\bk\le N-\bj} c(\bk,N-\bj)\,A^{\mathrm{SL}}_{\bk}.
\end{align*}
If $\bk\le\bj$, then $k_\ell\le\delta-1$ for $\ell\le f$, while
$\sum_{\ell>f}k_\ell\le \sum_{\ell>f}j_\ell=d-1$. Hence
$\bk\in\mathcal T_{d,\delta}$. By purity,
$A^{\mathrm{SL}}_{\bk}=0$ for all $\bk\neq \mathbf 0$, whereas
$A^{\mathrm{SL}}_{\mathbf 0}=K^2$. Therefore,
\begin{align*}
A^{\mathrm{U}}_{\boldsymbol j} = K^2 q^{-m} \binom{N}{\delta-1}^f \prod_{\ell=f+1}^s \binom{N}{j_\ell} .
\end{align*}
Similarly, $A^{\mathrm{U}}_{N-\boldsymbol j} \ge K^2 q^{-n+m}
\prod_{\ell=1}^f \binom{N}{\delta-1}\prod_{\ell=f+1}^s \binom{N}{j_\ell}$,
where the inequality follows from the positivity of the enumerators.
Substituting into \eqref{eq:equality}
yields $K \le q^{\,n-2m}$. \vspace{-0.15cm}
\end{proof}

\begin{remark} For general qLRCs, the expansion of
$A^{\mathrm U}_{\bj}$ may contain additional nonzero terms
$A^{\mathrm{SL}}_{\bk}$ with
$\bk\in\mathcal T_{d,\delta}\setminus\{\mathbf 0\}$.
These extra nonnegative contributions prevent the identity
$A^{\mathrm U}_{\bj}=K A^{\mathrm U}_{N-\bj}$
from yielding a bound on $K$.
\end{remark}\vspace{-0.1cm}

Next, we compare Theorem~\ref{thm:puresing} with the known bounds in the literature.
For disjoint $(r,\delta)$-qLRCs of length $n$ and minimum distance $d$, the quantum
dimension $k$ satisfies\vspace{-0.1cm}
\begin{align}
k \le n - 2(d-1) - 2\!\left(\frac{n}{N} - \Big\lceil \frac{d-1}{r}\Big\rceil\right)(\delta-1),
\label{slike}
\end{align}
which for $\delta=2$ reduces to the bound derived in~\cite{golowich2023quantum}. A direct comparison shows that \eqref{eq:singbound} strengthens the general disjoint
$(r,\delta)$-qLRC bound~\eqref{slike} under the additional purity assumption.

For pure stabilizer $(r,\delta)$-qLRCs, one also has\vspace{-0.1cm}
\begin{align}
   2d  \le  n -k - 2\Big( \Big\lceil\frac{n+k}{2r} \Big\rceil - 1 \Big)(\delta-1) + 2.
\label{singpurestab}
\end{align}
When $\delta=2$, Theorem~\ref{thm:puresing} agrees with the pure stabilizer
bound~\eqref{singpurestab}. For $\delta>2$, the comparison with~\eqref{singpurestab}
is more delicate because of the floor and ceiling terms, although the two bounds
coincide in several cases.

Theorem~\ref{thm:puresing} also matches the dimension bound obtained in the
dual-containing classical $(r,\delta)$-LRC setting. Indeed, the local-plus-global recovery argument guarantees correction of an erasure of size
$m=(d-1)+f(\delta-1)$, and hence the corresponding classical dimension satisfies
$k'\le n-m$. Substituting $k'=\frac{n+k}{2}$ yields exactly the bound in
Theorem~\ref{thm:puresing} for $K=q^k$. Thus the theorem recovers the corresponding
classical dimension bound in quantum form.

\section{Linear Programming Bound}
In this section, we derive an upper bound on the dimension of a pure disjoint
$(r,\delta)$-qLRC using an LP approach based on the blockwise weight enumerators
from Theorem~\ref{thm:LRCweight}.

\begin{theorem}[LP--based upper bound for pure disjoint $(r,\delta)$-qLRCs]
\label{thm:upper}
Let $\mathcal Q$ be a pure disjoint $(r,\delta)$-qLRC with parameters
$((n,K,d))_q$, and blockwise SL weight enumerators $\{A^{\mathrm{SL}}_{\bi}\}$ and $\{B^{\mathrm{SL}}_{\bi}\}$.
Let $f(\bx)$ be a polynomial with Krawtchouk expansion~\eqref{eq:expansion} such that
$f_{\bi}\ge 0$ for all $\bi\in\mathcal I_N^s$, $f_{\mathbf 0}>0$, and
$f(\bi)\le 0$ for all $\bi\in\mathcal T_{d,\delta}^{\mathrm{c}}$, where
$\mathcal T_{d,\delta}^{\mathrm{c}}:=\mathcal I_N^s\setminus\mathcal T_{d,\delta}$.
Then
\begin{align*}
K \le \frac{1}{q^n}\frac{f(\mathbf 0)}{f_{\mathbf 0}}.
\end{align*}
\end{theorem}
\begin{proof}
The claim follows from the LP method of~\cite{ashikhmin1999upper}, applied to
the blockwise weight enumerators in Thm.~\ref{thm:LRCweight} and 
Def.~\ref{def:pure}.
% It follows by the LP method in
% % The bound follows by applying the LP-based method of
% \cite{ashikhmin1999upper} applied to the blockwise weight enumerators in Thm.~\ref{thm:LRCweight} and Def.~\ref{def:pure}.
% that for pure codes,
% $A^{\mathrm{SL}}_{\boldsymbol i}=0$ for all
% $\boldsymbol i\in \mathcal{T}_{d,\delta}\setminus\{\boldsymbol 0\}$.
\end{proof}

% \begin{remark}
% For general $(r,\delta)$-qLRCs, one should impose the additional constraint 
% $f_{\boldsymbol{i}} > 0 \quad \forall \boldsymbol{i} \in \mathcal{T}_{d,\delta}$. Then, \vspace{-0.15cm}
% \begin{align*}
% K \le \frac{1}{q^n} \max_{\boldsymbol{i} \in \mathcal{T}_{d,\delta}} \frac{f(\boldsymbol{i})}{f_{\boldsymbol{i}}}.    
% \end{align*}
% \end{remark}
Note that the set of uncorrectable blockwise erasure patterns is $\mathcal T_{d,\delta}^{\mathrm{c}}=
\big\{\boldsymbol{i}\in\mathcal{I}_N^s \,\big|\, \sum_{j:\,i_j\ge\delta} i_j \ge d\big\}$.

\begin{theorem}[LP--based Singleton-like bound for pure disjoint $(r,\delta)$-qLRCs]\label{thm:sing1}
Let $\mathcal{Q}$ be a pure disjoint $(r,\delta)$-qLRC of length $n$ and minimum distance $d$. Then,
\begin{align}
    K \le \begin{dcases}
        q^{r(s- 2t) -s(\delta-1) + 2(\delta -\partial)} & \text{if}\;\; \partial \ge \delta, \\
        q^{r(s-2t) - s(\delta-1) } & \text{if}\;\; \partial < \delta,
        \end{dcases}
        \label{lpbound}
\end{align}
where $n=sN$ and $d = tN + \partial$  with $1\le \partial\le N$.
\end{theorem}

\begin{proof}[Proof sketch]
% Let $n=sN$ and write $d=tN+\partial$, with $1\le \partial\le N$. 
For $\partial \ge \delta$, define, following~\cite{agarwal2018combinatorial},
\begin{align*}
f(\bx)
&=
\prod_{j=1}^{s-t-1}
\left[
q^{2r}\prod_{i=\delta}^{N}\left(1-\frac{x_j}{i}\right)
\right]
q^{2(r+\delta-\partial)}
\prod_{i=\partial}^{N}\left(1-\frac{x_{s-t}}{i}\right).
\end{align*}
For $\partial < \delta$, define 
\begin{align*}
f(\bx)
&=
\prod_{j=1}^{s-t}
q^{2r}\prod_{i=\delta}^{N}\left(1-\frac{x_j}{i}\right).
\end{align*}
Standard Krawtchouk identities imply $f_{\bi}\ge 0$ for all $\bi$. It remains to verify that $f(\bi)\le 0$ for all $\bi\in \mathcal T_{d,\delta}^{\mathrm{c}}$. Suppose first that $\partial\ge\delta$ and let $\bi\in \mathcal T_{d,\delta}^{\mathrm{c}}$. If $i_\ell\ge\delta$ for some $1\le \ell\le s-t-1$, then one of the first factors vanishes. Otherwise, $i_1,\dots,i_{s-t-1}<\delta$, so these coordinates do not contribute to $\sum_{j:\,i_j\ge\delta} i_j$. Since the coordinates $s-t+1,\dots,s$ contribute at most $tN$, the condition $\sum_{j:\,i_j\ge\delta} i_j\ge d=tN+\partial$ forces $i_{s-t}\ge \partial$, and hence the last factor vanishes. 

Now suppose that $\partial<\delta$ and let $\bi\in \mathcal T_{d,\delta}^{\mathrm{c}}$. If $i_1,\dots,i_{s-t}<\delta$, then only the coordinates $s-t+1,\dots,s$ contribute to $\sum_{j:\,i_j\ge\delta} i_j$, so $\sum_{j:\,i_j\ge\delta} i_j \le tN < tN+\partial = d$, a contradiction. Therefore some $i_\ell\ge\delta$ with $1\le \ell\le s-t$, and one of the factors vanishes.
Thus $f(\bi)=0$ for all $\bi\in \mathcal T_{d,\delta}^{\mathrm{c}}$. Therefore $f$ satisfies the conditions of Theorem~\ref{thm:upper}, which yields the stated bound.
\end{proof} \vspace{-0.1cm}
To compare Theorem~\ref{thm:sing1} with the bounds
\eqref{slike}, \eqref{singpurestab}, and \eqref{eq:singbound}, we relax the floor
and ceiling terms in a way that preserves the direction of the inequalities.
Relaxing the floor in \eqref{eq:singbound} gives
\begin{align}
\hspace{-0.325cm}k \le r(s - 2t) - s(\delta-1) +2(\delta-\partial)
+ \frac{2(\partial-1)(\delta-1)}{N}.
\label{eq:first}
\end{align}
Relaxing the ceiling in \eqref{singpurestab} gives
\begin{align}
k \le r(s - 2t) - s(\delta-1) +\frac{2r(\delta-\partial)}{N}.
\label{eq:third}
\end{align}
Likewise, relaxing the ceiling in \eqref{slike} gives
\begin{align}
k \le\;& r(s - 2t) - s(\delta-1) + 2(\delta - \partial) \nonumber\\
&\quad + 2\!\left(\frac{t(\delta-1) + (\partial - 1)}{r}\right)(\delta-1).
\label{eq:second}
\end{align}
Theorem~\ref{thm:sing1} improves on the relaxed bounds
\eqref{eq:first}, \eqref{eq:third}, and \eqref{eq:second}. If $\partial\ge\delta$,
then \eqref{eq:first} and \eqref{eq:second} each contain an additional
nonnegative term beyond the LP exponent, and are therefore weaker. The bound
\eqref{eq:third} is also weaker since $N\ge r$. If $\partial<\delta$, then each
relaxed bound still contains a strictly positive excess term relative to the LP
bound. Hence, the LP bound is strictly stronger than these relaxed bounds in both
regimes.

\begin{remark}
The original bounds exhibit the same qualitative behavior. In the parameter
regimes considered here, the LP bound \eqref{lpbound} coincides with the pure
stabilizer bound \eqref{singpurestab} and is at least as tight as the pure
disjoint bound \eqref{eq:singbound}.
\end{remark}

\section{Conclusion}
Quantum error correction is central to quantum communication and many other
quantum technologies, and is essential for unlocking the potential of future
quantum-enabled communication systems~\cite{schwenteck20236g,AMIRI2026401,AMIRI2026423}.
In this paper, we studied pure disjoint $(r,\delta)$-qLRCs. We formulated local
KL conditions and introduced blockwise weight enumerators that track how error
weight is distributed across recovery sets. This framework yields a
Singleton-like bound that strengthens the known bound for disjoint
$(r,\delta)$-qLRCs under the purity assumption, as well as an LP-based upper
bound on the code dimension. These results provide a non-stabilizer,
weight-enumerator-based approach to the study of pure disjoint
$(r,\delta)$-qLRCs. Future work includes extending these techniques to impure
quantum codes and to $(r,\delta)$-qLRCs with overlapping recovery sets, as in
the classical setting~\cite{gruica2023lrcs}.

\newpage

\bibliographystyle{ieeetr}
\bibliography{ref}
\end{document}